%% file: mainobsstud.tex
\documentclass[twoside,11pt]{article}
\usepackage{obs_study_style}

\usepackage[utf8]{inputenc}

\usepackage[all]{xy}
\usepackage{amsmath}
\usepackage{csquotes}
\usepackage{subfig}
\usepackage{enumerate}
\usepackage{enumerate}
\usepackage{color}
\usepackage{mwe}
\usepackage{appendix}
\usepackage{centernot}
\usepackage{url,hyperref,xspace}
\usepackage[ruled,vlined]{algorithm2e}

\usepackage{soul}
\usepackage{xcolor}
\def\@centernot#1#2{%
  \mathrel{%
    \rlap{%
      \settowidth\dimen@{$\m@th#1{#2}$}%
      \kern.5\dimen@
      \settowidth\dimen@{$\m@th#1=$}%
      \kern-.5\dimen@
      $\m@th#1\not$%
    }%
    {#2}%
  }%
}
\makeatother

\newcommand{\bm}[1]{\mbox{$\mathbf #1$}}
\newcommand{\bY}{\bm{Y}}

\newcommand{\ie}{{\em i.e.\/}\xspace}
\newcommand{\eg}{{\em e.g.\/}\xspace}

\newcommand{\etc}{{\em etc.\/}\xspace}

\renewcommand{\eqref}[1]{\mbox{(\ref{eq:#1})}}

\newcommand{\secref}[1]{\mbox{\S$\,$\ref{sec:#1}}}

\newcommand{\tabref}[1]{\mbox{Table~\ref{tab:#1}}}

\newcommand{\thmref}[1]{\mbox{Theorem~\ref{thm:#1}}}

\newcommand{\corref}[1]{\mbox{Corollary~\ref{cor:#1}}}

\newcommand{\remref}[1]{\mbox{Remark~\ref{rem:#1}}}

\newcommand{\exref}[1]{\mbox{Example~\ref{ex:#1}}}

\newcommand{\Secref}[1]{\mbox{Section~\ref{sec:#1}}}

\newcommand{\setto}{\leftarrow}
\newcommand{\defeq}{:=}
\newcommand{\half}{\frac 1 2}

\newcommand{\ite}{{\rm ITE}}
\newcommand{\ate}{{\rm ATE}}
\newcommand{\cate}{{\rm CATE}}
\newcommand{\E}{{\rm E}}
\newcommand{\pb}{{\rm PB}}
\newcommand{\ph}{{\rm PH}}
\newcommand{\etal}{{\em et al.\/}}

\newtheorem{cor}{Corollary}

\heading{}{}{}{}{}{A.~Philip~Dawid and Stephen~Senn}
\ShortHeadings{Personalised Decision-Making}{Dawid and Senn}

\firstpageno{1}

\begin{document}

\title{Personalised Decision-Making without Counterfactuals}
\author{\name A.~Philip~Dawid \email apd@statslab.cam.ac.uk\\
  \addr Statistical Laboratory, University of Cambridge
  \AND
  \name Stephen~Senn \email stephen@senns.uk\\
\addr Statistical Consultant, Edinburgh
}

\date{Draft of \today}

\maketitle
\begin{abstract}%
  This article is a response to recent proposals by Judea
  Pearl and others for a new approach to personalised treatment
  decisions, in contrast to the traditional one based on statistical
  decision theory.  We argue that this approach is dangerously
  misguided and should not be used in practice.
\end{abstract}

\begin{keywords}
  decision theory,
counterfactual,
potential response,
preferred treatment
\end{keywords}

\section{Introduction}
\label{sec:intro}
\input{introduction2}


\section{The approach of Mueller and Pearl
}
\label{sec:approach}
\input{approach2}

\section{Decision-theoretic approach}
\label{sec:DT}
\input{DT2}

\section{Analysis}
\label{sec:analysis}
\input{analysis2}

\section{Comments on the approach}
\label{sec:comments}
\input{comments2}


\section{Examples}
\label{sec:ex}
\input{ex2}

\section{Another approach}
\label{sec:super}
\input{super2}






\section{Assumptions and critical comments}
\label{sec:prag}
\input{prag2}

\section{Discussion}
\label{sec:disc}
\input{disc2}

\section*{Acknowledgments}
We have benefited greatly from discussions with Mats Stensrud and
Aaron Sarvet.
 

\bibliography{strings,causal,references}


\end{document}

%% file: introduction2.tex
In recent works
\citep{mueller/pearl:blog,mueller/pearl:intro,li/pearl:ijcai19,li:22},
Judea Pearl and collaborators have set out an approach to personalised
treatment that is radically different from that based on traditional
statistical decision theory.  It is based on the conception that we
should care, not only about the outcome that actually materialises,
but also about the (necessarily unobserved, counterfactual) outcome
that, it is supposed, would have occurred under the treatment that was
not applied.  A similar conception forms the basis of other recent
work \citep{imai:jrssa,richens,benmichael}.  For earlier work in the
same vein, see \citet{gadbury/iyer:bcs, gadbury/iyer/albert:jspi}.

We consider this approach to be dangerously misguided, and believe
that real harm will ensue if it is applied in practice.  We argue our
case from a number of different viewpoints, and, expanding on the
comments of \citet{apd:imaidisc,sarvet:imaidisc}, explain why this
approach should not be regarded as a viable alternative to standard
statistical decision theory.

\subsection{Basic set-up}
The context is that of a ``target'' patient suffering from a disease,
for which a treatment is available.  The treatment is far from
perfect, so that not all treated patients recover, while some
untreated patients may recover anyway.  There is information available
on recovery rates for treated and untreated patients; both these rates
may depend on measured individual patient characteristics.  The basic
problem is to decide, on the basis of the target patient's own
characteristics, whether or not to treat him.  A variation is how to
prioritise patients for treatment when there are limited doses
available.

We introduce notation as follows:
\begin{description}
\item[Treatment decision] Binary decision variable $X$, coded 1 for
  treat, 0 for don't treat
\item[Response] Binary stochastic variable $Y$, coded 1 for recovery,
  0 for no recovery
\item[Individual background characteristics] Stochastic covariate $L$,
  possibly multivariate
\end{description}
We assume that that $L$ is determined and potentially observable prior
to any treatment decision, so that, in particular,
\begin{equation}
  \label{eq:l}
  \Pr(L=l \mid X \setto x) = \Pr(L=l).
\end{equation}
(Here $X \setto x$ denotes that an external intervention sets $X$ to
value $x$.)

For simplicity, we suppose that we can estimate, essentially
perfectly\footnote{\eg, from substantial data on $(L,X,Y)$, from
  experimental (or unconfounded observational) studies on patients
  whom we can regard as similar to the target (where the intuitive
  idea of {\em similarity\/} can be made mathematically precise in
  terms of {\em exchangeability\/}, as described for example in
  \citet{apd:found}).  \Secref{prag} below discusses some of the
  problems with this assumption in practice.},
the distribution of $Y$, conditional on $L$, under
either treatment intervention.  That is, we know the probabilities
\begin{eqnarray}
  \label{eq:p}
  p &:=& \Pr(Y=1 \mid L=l, X \setto 1)\\
  \label{eq:q}
  q &:=& \Pr(Y=1 \mid L=l, X \setto 0)
\end{eqnarray}
for any value $l$ of $L$.
The task is to use this
information to inform the treatment of the target patient.

\subsection{Outline}
\label{sec:outline}
We start in \secref{approach} with a brief outline of the the novel
approach proposed by Mueller, Pearl \etal.  This is contrasted, in
\secref{DT}, with a description of the traditional approach based on
statistical decision theory.  \Secref{analysis} considers the
Mueller/Pearl approach in more detail, following a logical path that
relates it to other problems, in particular the use of general
covariate information to strengthen conclusions, and introduces the
``preferred treatment'' covariate, whose probabilistic properties can
be identified by combining experimental and observational data; this
allows a deeper understanding of MP's analysis of this case.  We
follow this with some critical comments in \secref{comments}. In
\secref{ex} we give critical consideration to some examples from
\citet{mueller/pearl:intro}.  In \secref{super} we show how
information extracted from combining experimental and observational
data can be leveraged in the decision-theoretic framework to improve
personalised decision making.  \Secref{prag} notes some important
assumptions that are implicitly made in the analysis, and points out
that they are unlikely to hold in practice (although, purely for
argument's sake, we have largely accepted those assumptions prior to
this point).  \Secref{disc} summarises our analysis and conclusions.


%% file: approach2.tex
\subsection{Potential outcomes}
\label{perf}
The approach of \citet{mueller/pearl:intro} (henceforth MP) is based
on the concept of {\em potential outcomes\/} \citep{dbr:jep}.  This
conceives of the existence, even prior to treatment choice, of the
pair of variables $\bY = (Y(1), Y(0))$, where $Y(x)$ is intended to
denote the value that, it is supposed, $Y$ would take if an
intervention $X \setto x$ were to be applied.  The collection of
pretreatment variables $(Y(1), Y(0), L)$ is supposed to have a joint
distribution, necessarily unaffected by treatment.  With this
interpretation, we have
\begin{equation}
  \label{eq:interp}
  \Pr(Y(x)=y, L=l) = \Pr(Y=y, L=l \mid X\leftarrow x)
\end{equation}
whence, by \eqref{l},
\begin{equation}
  \label{eq:interpcond}
  \Pr(Y(x)=y \mid L=l) = \Pr(Y=y \mid  L=l, X\leftarrow x).
\end{equation}

In this approach, inference is ideally focused on the ``individual
treatment effect'',
\begin{equation}
  \label{eq:ite}
  \ite := Y(1)-Y(0),
\end{equation}
which can take values $+1$ (when $Y(1) = 1, Y(0)=0$: treatment
benefits the patient), $-1$ (when $Y(1) = 0, Y(0)=1$: treatment harms
the patient) or $0$ (treatment has no effect).  However, $\bY$, and
thus $\ite$, is never directly observable: this has been termed ``the
fundamental problem of causal inference'' \citep{pwh:jasa}.  Not only
can we not know a patient's \ite\ before the treatment decision is
made, we can not even know it later, when the outcome $Y$ is observed.
For if we treat the patient we will observe $Y = Y(1)$, but can not
then observe the counterfactual outcome $Y(0)$ relevant when we don't
treat; similarly, for an untreated patient we can observe $Y(0)$, but
not $Y(1)$.  So $\ite$ is always unobservable.  It follows that,
whatever data are available, no direct information can be gained about
the dependence between $Y(1)$ and $Y(0)$\footnote{Some indirect
  information arises from the fact that mere knowledge of marginal
  distributions imposes some inequality constraints on joint
  probabilities---see \secref{simple} below.}, and thus the full
distribution (marginal, or conditional on $L$) of $\ite$ can not be
estimated.  All that can be inferred is the conditional\footnote{as
  well as, of course, the unconditional expectation,
  $\ate=\E(\ite)= \E(\cate(L))$} expectation of $\ite$, the
``conditional average treatment effect'',
$\cate(l) :=\E(\ite \mid L = l)$, since, by \eqref{ite} and
\eqref{interp} with \eqref{p} and \eqref{q},
\begin{eqnarray*}
  \cate(l)   &=& \E\{Y(1) \mid L = l \}-\E\{Y(0) \mid L = l\}\\
             &=& \Pr(Y=1 \mid L=l, X \setto 1) - \Pr(Y=1 \mid L=l, X \setto 0)\\
             &=&p-q,
\end{eqnarray*}
which can thus be identified from experimental data.

In certain very special and highly atypical cases---essentially, those
where we have a fully deterministic and completely understood
mechanistic system governing the response to treatment---it could be
that the background knowledge $L$ is detailed enough to support
perfect prediction of the ensuing response, under either contemplated
intervention.  Then we will know, even in advance of treatment choice,
both potential outcome variables.  In this case $p = Y(1)$,
$q = Y(0)$, and $\cate(L) = \ite$.
However, in normal cases such perfect prediction is impossible.  Then
it is arguable whether the potential responses, and so \ite, even have
any meaningful existence.

\subsection{MP approach}
\label{sec:mp-approach}
MP's approach is based, first, on the idea that we would ideally like
to treat just those patients having $\ite=1$, for whom the treatment
makes a positive difference: they would not recover without it.  It
would be wasteful to treat a patient with $\ite = 0$, for whom the
treatment makes no difference, and harmful to treat a patient with
$\ite = -1$, who would recover if untreated, but not if treated.

However, this ideal is unattainable, as we can never know a patient's
$\ite$.  Attention is therefore diverted to the ``probability of
benefit'', $\pb: = \Pr(\ite=1) = \Pr(Y(1)=1, Y(0)=0)$, and the
``probability of harm'',
$\ph: = \Pr(\ite=-1) = \Pr(Y(1)=0, Y(0)=1)$.\footnote{When additional
  covariate information is available for the target case, these
  probabilities should be further conditioned on that} MP's second
suggestion, therefore, is that $\pb$ and $\ph$ should constitute the
main criteria by which (somehow) to assess and compare treatment
strategies.

But, even with extensive data, on account of the fundamental problem
of causal inference it is not generally possible fully to identify
$\pb$ and $\ph$.  Such data can, however, be used to set interval
bounds on these quantities.  It is MP's third fundamental
argument---the only one of the three that is applicable in
practice---that this interval information should (somehow) constitute
the basis for decision-making.

As an additional wrinkle (that is however inessential to the main
argument), MP point out that combining experimental and observational
data can narrow the inferred interval bounds on \pb\ and \ph.  In
certain very special cases these bounds can narrow to a single point,
leading to full identification of \pb\ and \ph, but this is
exceptional.

For further description, analysis and critique of the MP approach, see
\secref{analysis} below.


%% file: DT2.tex
We contrast the above with the traditional decision-theoretic (DT)
approach to treatment selection, which does not involve any
consideration of potential outcomes.

\subsection{Single patient decision problem}
\label{sec:single}

Consider first the case of a single target patient.  We have to decide
whether to offer this patient treatment, or not.

Having access only to the target patient's covariate value $L=l$, our
objective is to choose the treatment that will maximise the patient's
probability of recovery.  We should thus treat this patient if
$p= \Pr(Y=1 \mid L=l, X \setto 1) > q=\Pr(Y=1 \mid L=l, X \setto 0)$.
That is, we should treat just when $\cate(l) = p-q > 0$.\footnote{In
  the more general case that the outcome $Y$ is not binary, for
  example a survival time, we would need to associate a utility $U(y)$
  with the outcome $y$, and then treat the patient just when
  $\E\{U(Y) \mid L=l, X \setto 1)\} > \E\{U(Y) \mid L=l, X \setto
  0)\}$.  Applied to the binary case this reduces to the prescription
  above (so long as $U(1)> U(0)$).  In \citet{sarvet22} the above is
  termed the {\em interventionist\/} utility and approach, and
  contrasted with the {\em counterfactual\/} utility and approach
  implicit in \citet{mueller/pearl:intro} and explicit in
  \citet{benmichael,richens}.  Here we restrict to the binary case and
  do not use utilities.}

Faced with a large collection $G$ of patients to treat, and unlimited
supplies of the treatment, managing each patient (each with their own
value $l$ of $L$) according to the above rule will maximise the number
of recoveries.  That is, any other (deterministic or randomised)
decision rule that uses only the information on $L$ would lead to
fewer recoveries.

\begin{example}
  \label{ex:simple}
  Consider a case where
  \begin{eqnarray*}
  \Pr(Y=1 \mid  L=l, X\setto 1) &=& 0.49\\  
  \Pr(Y=1 \mid  L=l, X\setto 0) &=& 0.21.
  \end{eqnarray*}
  The conditional average treatment effect is
  $\cate(l) = 0.49 - 0.21 = 0.28$.  Since $\cate(l) > 0$, the optimal
  action is to treat this patient.  If we have a large collection of
  similar patients, with the same value $L=l$, they should all be
  treated---in which case the overall proportion of recovered patients
  will be $49\%$.  This is the best outcome that can be achieved by
  any treatment strategy: if we treated only a fraction $\alpha < 1$
  of these patients, the proportion recovering would be
  $\alpha\times 0.49 + (1-\alpha) \times 0.21 < 0.49$.
\end{example}




\subsection{Unit selection}
\label{sec:unitsel}
Consider now a large collection $G$ of patients $i=1,\ldots,N$, with
individual recovery probabilities
$p_i = \Pr(Y_i = 1 \mid L_i = l_i, X_i \setto 1)$,
$q_i = \Pr(Y_i = 1 \mid L_i = l_i, X_i \setto 0)$.  If we treat just
those in a subset $S$, the expected, and asymptotically actual,
proportion recovering will be
$N^{-1}\{\sum_{i\in S} p_i + \sum_{i \in G \setminus S} q_i\} =
N^{-1}\{\sum_{i\in G} q_i + \sum_{i\in S} \cate_i\}$.  Consequently,
to maximise this proportion, we should choose $S$, subject to any
constraints, to maximise $\sum_{i\in S} \cate_i$.  If we have limited
treatments available, we should thus prioritise individuals in
decreasing order of their \cate\ (while of course not treating any one
for whom $\cate<0$.)  Any other policy subject to the same constraints
will have a smaller number of recoveries.

\subsection{Missing information}
\label{sec:missinf}
It may happen that, while we have full knowledge of the distribution
of $(Y,L)$ conditional on each intervention $X\setto x$ ($x=1,0$), the
value of $L$ for the target patient is not available.  In that case we
can not proceed conditionally on $L$, and have no option but to base
the management of the patient on the unconditional\footnote{more
  generally, conditional on whatever limited background information
  {\em is\/} available} probabilities $\Pr(Y=1 \mid X \setto x)$.
Nothing is gained by trying to impute the unknown value of $L$ for
this patient.  If this is not obvious (as it should be), suppose we
tried to do so.  The recovery probabilities, conditional on a
hypothesised value $l$ for $L$, are $\Pr(Y=1 \mid L=l, X \setto x)$
($x=1,0$).  But as we do not know $l$, we need to take the expectation
of $\Pr(Y=1 \mid L, X \setto x)$ over the known distribution of $L$
when setting $X\setto x$ (which by \eqref{l} does not depend on $x$).
But this is just the unconditional recovery probability
$\Pr(Y=1 \mid X \setto x)$.  Wishful thinking is delusional: we can
only ever work with what we do in fact know---not with what we would
like to have known, if only circumstances had been different.

\subsubsection{Relation with MP}
\label{sec:rel}
To relate the above to the MP framework, with its assumed potential
responses, suppose we were to take, as our missing covariate $L$, the
individual treatment effect $\ite$ (which in this approach is supposed
to exist and to be determined prior to actual treatment).  If we were
(ideally, though never actually) able to observe $\ite$ before
deciding how to treat a patient, we would know exactly what to do:
treat just when $\ite=1$.  However, as discussed above, in the light
of the necessarily absent knowledge of $L=\ite$, the mere existence of
this ideal situation can have no bearing on the actual decision
problem at hand.  Nothing would be gained by trying to impute $\ite$:
we must again simply focus on $\cate = p-q$.  Consideration of
potential responses (even if regarded as meaningful) does not add any
value to the decision-theoretic approach.



%% file: analysis2.tex
We now provide a more detailed description and deconstruction of the
analysis of MP~\citep{mueller/pearl:intro}---which should not,
however, be taken as agreement with their arguments and
interpretations.  There are a number of crucial assumptions required,
but to avoid cluttering the argument we postpone specification and
discussion of these to \secref{prag}.

We develop the storyline in a number of stages.

In \secref{simple} we consider the case where we have access to
extensive experimental data on treatment $X$ and response $Y$, and
show how this can be used to bound the probabilities of benefit and of
harm.  We also discuss the special circumstances in which these
interval bounds shrink to a point.

In \secref{cov} we further suppose that we can measure covariate
information $L$ on individuals.  If we have this additional
information for the experimental data, this can lead to a narrowing of
the bounds for $\pb$ and $\ph$ for the target case---even when the
value of $L$ for that case is unobserved.

\Secref{itt} introduces a particular, potentially useful, covariate,
the ``preferred treatment'', $X^*$---the treatment that a patient (or
their doctor) would ideally choose, if unconstrained.  This may well
be informative about their state of health, and thus their outcome.
In some experiments it may be possible to observe $X^*$, and this can
then be used directly as $L$ in \secref{cov}.  However it will often
not be possible to observe $X^*$ in the experiment.  \Secref{comb}
considers how this problem can be overcome by the incorporation of
observational data, if we can assume that, in such data, the preferred
treatment was the one actually applied, so that $X^* = X$ becomes
observable.  The combination of experimental and observational data
then identifies the distribution of $X^*$ (together with the other
variables), so allowing us to apply the theory of \secref{cov} to
obtain improved bounds for $\pb$ and $\ph$, as detailed in
\secref{app}.

\subsection{Simplest case}
\label{sec:simple}
We start by presenting the basis of the approach in the simplest case,
where the data are purely experimental, and there is no additional
covariate information.  We then have access just to the interventional
response probabilities forming the basis of the decision-theoretic
analysis, $p=\Pr(Y=1 \mid X\leftarrow 1) = \Pr(Y(1)=1)$,
$q=\Pr(Y=1 \mid X\leftarrow 0) = \Pr(Y(0)=1)$.  But what can be
inferred, from these, about the probabilities of benefit and of harm?

As described by \citet{apd/mm:sef}, it is helpful to re-express the
interventional probabilities in terms of parameters $\tau$ and $\rho$,
where
\begin{eqnarray}
\label{eq:tau}  
\tau &\defeq& p-q\\
\label{eq:rho}  
  \rho &\defeq& p+q-1.  
\end{eqnarray}
Then $\tau$ is the {\em average treatment effect\/}, \ate, of $X$ on
$Y$, while $\rho$ is a measure of how common the outcome is.  Each of
$\tau$, $\rho$ can take values in $(-1,1)$.

The transition matrix from $X$ to $Y$, with entries
$\Pr(Y=y \mid X\setto x)$, is
\begin{equation}
  \label{eq:P}
  P =
  \left(
    \begin{array}[c]{cc}
      p &1-p\\
      q & 1-q
    \end{array}
  \right)
  =
  \left(
    \begin{array}[c]{cc}
      \half(1+\tau+\rho) & \half(1-\tau-\rho)\\
      \half(1-\tau+\rho) & \half(1+\tau-\rho)
    \end{array}
  \right) =:  P(\tau,\rho),
\end{equation}
where the row ($x$) and column ($y$) labels are $1$ and $0$ in that
order.  The necessary and sufficient condition for all the entries to
be non-negative is
\begin{equation}
  \label{eq:rhotau}
  |\tau| + |\rho| \leq 1.
\end{equation}
We have equality in \eqref{rhotau} only in the {\em
  degenerate\/} case that one of the entries of $P$ is 0.


We can express the joint distribution for $\bY=(Y(1),Y(0))$ as in
\tabref{r0r1}.
\begin{table}[h]
  \centering
  \begin{tabular}[c]{c|cc|c} 
    \multicolumn{1}{c}{} & {$Y(0)=1$} & \multicolumn{1}{c}{$Y(0)=0$}\\\hline
    $Y(1) = 1$  & $\half(1+\rho-\xi)$  & $\half(\xi+\tau)$ & $\half(1+\tau+\rho)$\\
    $Y(1) = 0$  & $\half (\xi-\tau)$ & $\half(1-\rho - \xi)$  & $\half(1-\tau-\rho)$\\
    \hline
                             & $\half(1-\tau+\rho)$  &  $\half(1+\tau-\rho)$& $1$\\
  \end{tabular}
  \caption{Joint probability distribution of $(Y(1),Y(0)$}   
  \label{tab:r0r1}
\end{table}
The margins are determined by \eqref{interp} (with $L$ absent),
\eqref{tau}, \eqref{rho} and \eqref{P}; but the internal entries are
indeterminate, having one degree of freedom crystallised in the
unspecified ``slack variable'' $\xi$, which, on account of the
fundamental problem of causal inference, is not identified by the
experimental data.  The only constraint on $\xi$ is the logical one
that all internal entries of \tabref{r0r1} be non-negative. This holds
if and only if
\begin{equation}
  \label{eq:basicineq}
  |\tau| \leq \xi \leq  1-|\rho|.
\end{equation}
This interval information is all that can be concluded about the joint
distribution for $\bY$ when we have data on the behaviour of $Y$ under
intervention on $X$, and no additional information.  

\begin{remark}
  \label{rem:pointsimp}
  The interval \eqref{basicineq} shrinks to a point, so that the joint
  distribution of $\bY$ is fully determined by the experimental data,
  if and only if we have equality in \eqref{rhotau}, \ie, just when
  $P$ is degenerate, so that, for some $x,y =0,1$,
  $\Pr(Y=y \mid X \setto x) =0$.  That is to say, for at least one of
  the interventions, the resulting outcome $Y$ can be predicted with
  certainty---a most unusual state of affairs.  In this case
  $\Pr(Y(x) = y) = 0$, so that both joint events
  $(Y(x) = y, Y({\overline x})=0)$ and $(Y(x) = y, Y(\overline x)=1)$
  (where $\overline x = 1-x)$ have probability $0$.
\end{remark}

\subsubsection{Benefit and harm}
\label{sec:benharm1}
The probability of benefit $\pb$ is the upper right entry of
\tabref{r0r1}, $\pb = \Pr(Y(1)=1,Y(0)=0) = \half(\xi+\tau)$, which by
\eqref{basicineq} is bounded between
\begin{equation}
  \label{eq:pb-}
  \pb_- := \half(|\tau|+\tau) = \max\{\tau,0\}
\end{equation}
and
\begin{equation}
  \label{eq:pb+}
  \pb_+:=\half(1-|\rho|+\tau) = \min\{\Pr(Y=1 \mid X \setto 1), \Pr(Y=0
  \mid X \setto 0)\}.
\end{equation}
The probability of harm, given by the lower left entry of
\tabref{r0r1}, is then
$\ph=\Pr(Y(1)=0,Y(0)=1) = \half(\xi-\tau) = \pb-\tau$.

For the case of \exref{simple}, we have $\tau= 0.28$, $\rho= -0.3$.
Without any further information, we can only infer
$0.28 \leq \pb \leq 0.49$, and correspondingly $0 \leq \ph \leq 0.21$.

\subsubsection{Biased data}
\label{sec:rest}
In some cases it may be that we have access to the value of $\tau$,
but not $\rho$: for example, we may consider that, because of
selection effects, the data supply biased estimates for $p$ and $q$,
but nevertheless the estimate of their difference, $\tau$, is
essentially unbiased and can be regarded as relevant to the target
patient.  In that case, applying \eqref{pb-} and \eqref{pb+} knowing
only that $0\leq |\rho|\leq 1 - |\tau|$, we can bound $\pb$ between
$\max\{\tau,0\}$ and $\half(1+\tau)$.\footnote{This situation is
  analysed in \cite{mueller/pearl:rct}.  However, the upper bound is
  given there as $\min\{1, 1+\tau\}$, which is worse than ours.}  For
\exref{simple}, we obtain $0.28 \leq \pb \leq 0.64$, and consequently
$0 \leq \ph \leq 0.36$.

\subsection{Covariate information}
\label{sec:cov}
Now suppose that, again with purely experimental data, we can obtain
additional information on some pre-treatment covariate information $L$
(for simplicity assumed discrete), where each term
$\Pr(L=l \mid X\leftarrow x) = \Pr(L=l)$ is known (and assumed
positive).  We thus have access to the conditional interventional
probabilities $\Pr(Y=y \mid L=l, X \leftarrow x)$.


Let $\tau(l), \rho(l)$ be defined as in \eqref{tau} and \eqref{rho},
but with probabilities further conditioned on $L=l$.  If, for the
target case, we observe $L=l$, then we simply apply the above
analysis, conditional on $L=l$.  In particular, the joint distribution
for $\bY$, given $L=l$, will be as in \tabref{r0r1}, with
$\rho,\tau,\xi$ replaced, respectively, by $\rho(l),\tau(l),\xi(l)$,
where $\xi(l)$ is subject only to
\begin{equation}
  \label{eq:sineq}
  |\tau(l)| \leq \xi(l) \leq  1-|\rho(l)|.
\end{equation}

Finally, suppose that, while still having access, from the
experimental data, to the probabilities
$\Pr(Y=y \mid L=l, X \leftarrow x)$, we do not observe $L$ for the
target patient.  In this case (and unlike the situation for decision
theory) this additional probabilistic knowledge can make a difference
to the MP inference.  In \tabref{r0r1} we now have
$\xi = \sum_s \xi(l)\times \Pr(L=l)$, and we get the new interval
bound
\begin{equation}
  \label{eq:impb}
  L := \sum_s |\tau(l)|\times \Pr(L=l) \leq \xi \leq 1-\sum_s |\rho(l)| \times \Pr(L=l)=:U.
\end{equation}

Since $\tau= \sum_s \tau(l)\times \Pr(L=l)$,
$\rho = 1-\sum_s \rho(l) \times \Pr(L=l)$, this interval will be
strictly contained in that of \eqref{basicineq}, so long as not all the
$(\tau(l))$, or not all the $(\rho(l))$, have the same sign.

The probability of benefit is now bounded below by
$\sum_l \pb_-(l)\Pr(L=l)$ and above by $\sum_l \pb_+(l)\Pr(L=l)$,
where $\pb_-(l)$ and $\pb_+(l)$ can be computed as in
\secref{benharm1} with $\tau$ and $\rho$ replaced by $\tau(l)$ and
$\rho(l)$, respectively.

\begin{remark}
  \label{rem:points}
  Applying \remref{pointsimp}, and noting $|\tau(l)|\leq 1-|\rho(l)|$,
  all $l$, we see that the interval \eqref{impb} will reduce to a
  point, yielding full identification of the joint distribution of
  $\bY$, if and only if $|\tau(l)| = 1-|\rho(l)|$, all $l$, so that,
  for each $l$, at least one of $\Pr(Y=y \mid L=l, X\setto x)$, for
  $x,y=0,1$, is zero.  In this case, both
  $\Pr(Y(x) = y, Y({\overline x})=0 \mid L=l)$ and
  $\Pr(Y(x) = y, Y(\overline x)=1\mid L=l)$ will be $0$.  Knowing the
  value $l$ of $L$ would then allow us to predict at least one of the
  interventional outcomes with certainty.  However, the relevant $x$
  and $y$ may vary with $l$, in which case such certainty would not be
  possible in the absence of knowledge of $L$.
\end{remark}

\subsubsection{Observational data}
\label{sec:obs}
Consider now the case that our data are observational, rather than
experimental.  Suppose we can observe a {\em sufficient covariate\/}:
a covariate $L$ such that after conditioning on $L$ there is no
residual confounding.  That is to say, the observational probability
$\Pr(Y =y \mid L=l, X = x)$ can be equated with the interventional
probability $\Pr(Y =y \mid L=l, X \setto x)$.  To ensure meaningful
conditioning, we further need the {\em positivity\/} condition: in the
observational setting,
\begin{equation}
  \label{eq:pos}
  \Pr(L=l, X= x) >0\qquad\mbox{all $l$, and $x=0$ or 1}.
\end{equation}
We can then proceed exactly as in \secref{cov} above.

\subsection{Preferred treatment}
\label{sec:itt}
Administration of a treatment to a patient can be usefully decomposed
into
two steps:\\

\begin{description}
\item[Preference] The patient, or their doctor, decides which
  treatment they would ideally want.  We introduce a binary stochastic
  ``preferred treatment''\footnote{In \cite{apd:found} this variable
    was termed ``intention to treat''.  We have avoided that usage
    here, to forestall confusion with the distinct use of that term in
    a clinical trial under non-compliance.} (PT) variable $X^*$ to
  denote this preferred
  treatment.\\
  
  Note that such a preference will typically be related to the
  patient's health status and other background information that could
  be predictive of recovery, so that we cannot regard those who
  desire, and those who reject, active treatment as comparable like
  with like.
  (This is the genesis of confounding.)\\
 
\item[Application] A treatment $X$ (which may or may not be the same
  as the preferred treatment $X^*$)\footnote{It is important to
    distinguish between $X$ and $X^*$.  We should ideally further
    distinguish between imposed treatment and received treatment, as
    in \citet{apd:found}.  Here we notate both as $X$, hoping this
    will cause no confusion.  We write $X\setto x$ when $X$ refers to
    the imposed treatment, and $X = x$ when $X$ refers to
    the received treatment.} is taken by the patient.\\
  
\end{description}

The PT variable $X^*$ is supposed to exist prior to application of
treatment, and can thus be regarded as independent of it:
\begin{equation}
  \Pr(X^* = x^* \mid X \leftarrow x) = \Pr(X^*=x^*).
  \label{eq:xint}
\end{equation}
This expresses the covariate nature of $X^*$.

We assume that, in an observational setting, the preferred treatment is
the one that is actually administered (there being no reason to do
otherwise).  Thus the received treatment $X$ will be the same as the
preferred treatment $X^*$.  In particular, since we observe $X$, we can
infer the value of $X^*$ in an observational context.

In an experiment, however, the treatment $X$ will be imposed (\eg, by
randomization), in a way that will typically take no account of $X^*$.
Even though we can still conceive of the PT variable $X^*$ as
existing, it may or may not be possible to observe it.  In the case
that $X^*$ is observed in the data---as in a patient preference trial
\citep{Torgerson:bmj})---it can be used, just like any other
covariate, to improve decision-making, as in \secref{DT} (when $X^*$
is observed for the target patient) \citep{stensrud2022optimal}; or, in
the approach of MP, to potentially narrow the bounds on $\pb$ and
$\ph$ even when $X^*$ is not observed for the
target, as in \secref{cov}.\\

\subsubsection{PT as a sufficient covariate}
\label{sec:ittsuff}
In an observational setting, where $X^*=X$ is observed, it is natural
to assume ``distributional consistency'' \citep{apd:found}: the
distribution of $Y$ given preferred treatment $X^*=x$---and so, also,
given received treatment $X=x$---is the same as that of $Y$, given
$X^*=x$, under an imposed intervention $X \leftarrow x$ that happens
to coincide with the treatment that would have been chosen anyway:
\begin{equation}
  \label{eq:distcons}
  \Pr(Y=y \mid X=x) = \Pr(Y=y \mid X^*=x,X = x) = \Pr(Y=y \mid X^*=x,X  \leftarrow x).
\end{equation}
For $x^*\neq x$, the event $(X^*=x^*,X = x)$ does not occur in the
observational regime, so we can interpret
$\Pr(Y=y \mid X^*=x^*,X = x)$ however we want, in particular as
\begin{equation}
  \label{eq:distcons2}
\Pr(Y=y \mid X^*=x^*,X = x) = \Pr(Y=y \mid X^*=x^*,X  \leftarrow x),
\end{equation}
and then \eqref{distcons} implies that \eqref{distcons2} holds for all
$x,x^*$.

Properties~\eqref{xint} and \eqref{distcons2} imply that $X^*$, which
is observed in the observational setting, behaves as a sufficient
covariate. 

\subsection{Combination of data}
\label{sec:comb}
It would be nice if, with observational data, we could profit from the
fact that $L = X^*$ is a sufficient covariate, as in \secref{obs}.
However, this is not straightforward, since the positivity condition
\eqref{pos} fails: for $x^*\neq x$, even though we may assume
$\Pr(Y=y \mid X^*=x^*, X\setto x) = \Pr(Y=y \mid X^*=x^*, X = x)$, we
have no data to estimate the latter term.  Again, when our data are
experimental but we can not directly observe $X^*$, we can not
identify $\Pr(Y=y \mid X^*=x^*, X\setto x)$.  However, it turns out
that we can do so if we can also obtain observational data: the
combination of both types of data allows us, after all, to identify
$\Pr(Y=y \mid X^*=x^*, X\setto x)$, even for $x\neq x^*$.  This we
show in the following theorem.

\begin{theorem}
\label{thm:dawid}
Suppose we can identify the joint distribution of $X$ and $Y$ in the
observational context, where $0<\Pr(X=1)<1$, and can also identify the
distribution of $Y$ under either intervention $X \leftarrow x$
$(x=0,1)$.  Then, under conditions \eqref{xint} and \eqref{distcons},
all the probabilities $\Pr(Y=y \mid X^*=x^*, X \leftarrow x)$
$(x,x^* = 0,1)$ are identified.  Specifically,
\begin{eqnarray}
  \label{eq:11}
  \Pr(Y=y \mid X^* = 1, X \leftarrow 1) &=& \Pr(Y=y \mid X=1)\\
  \label{eq:10}
  \Pr(Y=y \mid X^*=1, X \leftarrow 0) &=& \frac{\Pr(Y=y \mid  X \leftarrow 0) - \Pr(Y=y, X=0)}{\Pr(X=1)}\\
  \label{eq:00}
  \Pr(Y=y \mid X^* = 0, X \leftarrow 0) &=& \Pr(Y=y \mid X=0)\\
  \label{eq:01}
  \Pr(Y=y \mid X^*=0, X \leftarrow 1) &=& \frac{\Pr(Y=y \mid  X \leftarrow 1) - \Pr(Y=y, X=1)}{\Pr(X=0)}.
\end{eqnarray}
\end{theorem}

\begin{proof}
  \eqref{11}  and \eqref{00} follow from \eqref{distcons}.
  
  To identify $\Pr(Y=y \mid X^* = 1, X \leftarrow 0)$, we argue as
  follows.  We have
\begin{eqnarray}
\nonumber  \Pr(Y=y \mid X \leftarrow 0) &=& \Pr(Y=y \mid X^*=0,  X\leftarrow 0)\times\Pr(X^*=0 \mid  X \leftarrow 0)\\
\nonumber                               &&{}+\Pr(Y=y \mid X^*=1,  X\leftarrow 0)\times\Pr(X^*=1 \mid  X \leftarrow 0)\\
\nonumber                               &=& \Pr(Y=y \mid X=0)\times\Pr(X=0)\\
  \label{eq:follow}                               &&{}+\Pr(Y=y \mid X^*=1,  X\leftarrow 0)\times\Pr(X=1),
\end{eqnarray}
on using \eqref{xint} and \eqref{distcons}, and the fact that $X^*=X$
in the observational setting.  Since all the other terms in
\eqref{follow} are identifiable in either the observational or the
experimental context, and $\Pr(X=1)\neq 0$, we can solve for
$\Pr(Y=y \mid X^*=1, X\leftarrow 0)$, obtaining \eqref{10}.  Then
\eqref{01} follows similarly.
\end{proof}
The above proof relies on $X$ (and so $X^*$) being binary, but $Y$
need not be.  Versions of this argument have appeared in
\cite{forcinacomment,shpitser/pearl:ett,sgg/apd:ett,apd/mm:sef,stensrud2022optimal}.

\begin{cor}
\label{cor:joint}
The joint distribution of $(X^*, Y)$ under an intervention
$X \leftarrow x$ is then identified.
\end{cor}
\begin{proof}
  Follows since, by \eqref{xint},
  $\Pr(X^* = x^* \mid X \leftarrow x) = \Pr(X=x^*)$ is identified in
  the observational context.
\end{proof}

\begin{remark}
  \label{rem:constraint}
  Since $\Pr(Y=y \mid X^*=1, X\leftarrow 0)\geq 0$, \etc, we deduce
  from \eqref{10} and \eqref{01} the consistency constraint
  $\Pr(Y=y \mid X\leftarrow x) \geq \Pr(Y=y, X=x)$, all $x, y$.  When
  this fails, and that failure can not be ascribed to sampling
  variation or bias, that is evidence of violation of the conditions
  of \secref{prag} below, that have, implicitly, been used to justify
  the above argument.
\end{remark}

\thmref{dawid} and \corref{joint} clarify just what it is that the
combination of observational and experimental data is doing for us: it
allows us to identify distributions involving the PT variable $X^*$.

\subsection{Combination of data: Benefit and harm}
\label{sec:app}
Taking now $X^*$ as our sufficient covariate $L$, on combining
experimental and observational data we can apply the formulae of
\eqref{11}--\eqref{01} to compute the quantities $\tau(x^*)$,
$\rho(x^*)$ required for the analysis of \secref{cov}.  Noting that
$X^*=X$ in the observational regime, so that $\Pr(X=x)=\Pr(X^*=x)$, we
obtain
\begin{eqnarray*}
  \label{eq:tau1}
  {\Pr(X^*=1)} \times \tau(1) &=& {\Pr(Y=1) - \Pr(Y=1 \mid X \setto 0)}\\
  {\Pr(X^*=0)}\times\tau(0) &=& {\Pr(Y=1 \mid X \setto 1)-\Pr(Y=1)}\\
  {\Pr(X^*=1)} \times\rho(1) &=& K - \Pr(Y=0 \mid X \setto 0)\\ 
  {\Pr(X^*=0)}\times\rho(0) &=& \Pr(Y=1 \mid X \setto 1)-K
\end{eqnarray*}
where $K= \Pr(Y=1, X=1) + \Pr(Y=0, X=0)$.  Then from \eqref{impb} we
bound $\xi$ within $(L,U)$, where
\begin{eqnarray*}
  L &=& |\Pr(Y=1) - \Pr(Y=1 \mid X \setto 0)| + |\Pr(Y=1) - \Pr(Y=1 \mid X \setto 1)|\\
1-U &=& |\Pr(Y=0 \mid X \setto 0) - K| + |\Pr(Y=1 \mid X \setto 1) - K|.
\end{eqnarray*}
Then $\pb$ lies in $(\half(L+\tau), \half(U+\tau))$, and
$\ph = \pb-\tau$ lies in $(\half(L-\tau), \half(U-\tau))$.  Although
expressed differently, these results agree with those of MP.

By \remref{points}, the joint distribution of $\bY$, and in particular
$\pb$, $\ph$, will be point identified just when, for both $x^* = 0$
and $x^* = 1$, there exist $x$, $y$ such that
$\Pr(Y=y \mid X^* = x^*, X \setto x)= 0$.  In non-trivial cases we
will have $\Pr(Y = y \mid X=x) \neq 0$, in which case, by \eqref{11}
and \eqref{00}, this would need to happen with $x \neq x^*$.  For
that, by \eqref{10} and \eqref{01}, we require
\begin{equation}
  \label{eq:a10}
  \Pr(Y=y \mid X\setto 0) = \Pr(Y=y,  X=0) 
\end{equation}
for either $y=1$ or $y=0$; as well as
\begin{equation}
  \label{eq:a01}
  \Pr(Y=y \mid X\setto 1) = \Pr(Y=y,  X=1) 
\end{equation}
for either $y=1$ or $y=0$.

\subsection{Combination of data: DT approach}
\label{sec:combdt}
As discussed in \secref{missinf}, in the decision-theoretic approach,
when covariate data are missing there is no point in trying to impute
them.  The only ``value added'' by combining experimental and
observational data is knowledge of distributions involving the PT
covariate $X^*$.  But, in the absence of knowledge of the actual value
of $X^*$, this purely distributional knowledge has no bearing on the
decision problem.  Consequently, for the DT analysis, there is nothing
to be gained by bringing observational data into the picture.


%% file: comments2.tex

Here we collect some critical comments on the MP programme, with its
basis in potential responses.  These are arranged along several
dimensions.

\subsection{Necessity?}
\label{sec:nec}
The potential response framework, introduced by \citet{jn:ss} and
reintroduced and popularised by \citet{dbr:jep}, is often taken to be
fundamental to the conduct of causal inference.  However this
assumption has been refuted by \citet{apd:cinfer,apd:found}, who by
detailed analysis has shown that, so far from being fundamental,
potential responses are unnecessary and unhelpful.  A fully
satisfactory and straightforward theory of statistical causality,
ideally suited for making statistical inferences about the ``effects
of causes'' \citep{apd/mm:eoccoe}, as is relevant here, can be
constructed using only standard decision-theory, eschewing any
additional ingredients such as potential responses.

\subsection{Philosophy}
\label{sec:phil}
There are also serious philosophical objections to regarding potential
responses as having real existence.  Only if we take a fully Laplacean
view of the universe, in which the future of the universe is entirely
determined by its present state and the laws of Physics, does this
make any sense at all---and even then, it is difficult to incorporate
the whims of an unconstrained external agent who decides whether or
not to give treatment, or to account for the effect of external
conditions arising after treatment.

Even under Laplacean determinism, our ignorance of the information
needed to predict the future exactly means that we are unable to make
use of it.  Whether or not we believe in a deep-down deterministic
universe, our predictions of the future can only be based on the
limited information we do have at our disposal, and must necessarily
be probabilistic.\footnote{See \citet{apd:statsci04} for an
  approach to understanding non-extreme probabilities based on
  imperfect information about a deterministic world.} Imagining what
we could know or do, if only we had more information than we actually
do have, gets us nowhere.

\subsection{Applicability}
\label{sec:applic}
Another important dimension of criticism is that the strong conditions
needed for application of the MP theory will almost never obtain in
practice.  See \secref{prag} below for details.

\subsection{Helpfulness}
\label{sec:help}
The output of an MP analysis will, at the very best (and then only
rarely), be point estimates of the probabilities of benefit and of
harm.  In typical cases we can only bound these quantities within an
interval.  But even when we know these quantities, it is far from
clear how they should be used directly to inform treatment decisions.

In the examples of \citet[Section 3]{mueller/pearl:intro} (see also
\secref{ex} below), the experimental data do not distinguish between
males and females, but the incorporation of observational data shows
up differences in their estimates for \pb\ and \ph.  MP suggest this
difference indicates that there should be further variables that could
be measured that would lead to different decision strategies for the
two sexes, but give no guidance for identifying such variables.  The
general point is valid, and indeed such guidance is available: see
\secref{super} below.  But it is considerations other than \pb\ and
\ph\ that are relevant there.

\citet{MP:aje} consider the (unusual) situation that ``a post-mortem
autopsy can identify the cause of death and that families of patients
who die due to the treatment are likely to sue the hospital for
negligence.''  That is to say, at post-mortem new information becomes
available that shows that the patient would have lived if untreated,
and so has suffered harm.  This is proferred as an argument for paying
attention to benefit and harm.  But the defence against such a law
suit would be that this information could not possibly have been known
or taken into account at the point of treatment, so there can be no
culpability.  This is yet another example of the fact there is nothing
to be gained by attempting to impute unobserved quantities.

\subsection{Ethics}
\label{sec:ethics}
Our final criticism is the simplest, but most incisive.  The treatment
decisions made using the DT approach are guaranteed to be better than
those made by any other decision rule, in the sense that they will
maximise the number of recoveries in the population.  So whenever some
other approach, such as that of MP, leads to different decisions, it
will produce a decrease in the number of recoveries.  We find it hard
to construe this as ethical.


%% file: ex2.tex
\citet[Table~1]{mueller/pearl:intro} consider two cases, in both of
which the interventional probabilities of recovery are as in our
\exref{simple}, having $\Pr(Y=1 \mid X\setto 1) = 0.49$,
$\Pr(Y=1 \mid X\setto 0) = 0.21$, and so $\ate = 0.28$.  However, they
have different observational data.  We now analyse these in detail.

\begin{example}
  \label{ex:f}
  This example relates to females, for whom the observational joint
  probabilities are as in \tabref{f}.
\begin{table}[h]
  \centering
  \begin{tabular}[c]{c|cc|c} 
    \multicolumn{1}{c}{} & {$Y=1$} & \multicolumn{1}{c}{$Y=0$}\\\hline
    $X =1$ & $0.19$ & $0.51$ & $0.70$\\
    $X = 0$ & $0.21$ & $0.09$ & $0.30$\\\hline
                         &  $0.40$ & $0.60$ & $1$\\
  \end{tabular}
  \caption{Joint observational distribution of $(X,Y)$ for females}   
  \label{tab:f}
\end{table}

Applying the formulae of \secref{app} we find:
\begin{eqnarray*}
  0.7 \times\tau(1) &=& 0.19\\
  0.3 \times\tau(0) &=& 0.09\\
  0.7 \times\rho(1) &=& -0.51\\
  0.3 \times\rho(0) &=& 0.11
\end{eqnarray*}
It follows that $\pb_-(1) = \tau(1) = 19/70$.  Also,
$\pb_+(1) = \Pr(Y=1 \mid X^* =1, X\setto 1) = \Pr(Y=1 \mid X=1) =
19/70$.  Hence, given $X^* = 1$, we have exact identification:
$\pb(1) = 19/70$.  This occurs because
$\Pr(Y=1, X=0) = 0.21 = \Pr(Y=1 \mid X\setto 0)$, implying the
deterministic property $\Pr(Y=1 \mid X^* = 1, X \setto 0) = 0$: a
female who desires treatment will never recover if untreated.
Consequently such a female should be treated.

Also, $\pb_-(0) = \tau(0) = 0.3$, while
$\pb_+(0) = \Pr(Y=0 \mid X^* =0, X\setto 0) = \Pr(Y=0 \mid X=0) =
0.3$.  Given $X^* = 0$, we again have exact identification:
$\pb_-(0) = 0.3$.  This occurs because
$\Pr(Y=0, X =1) = 0.51 = \Pr(Y=0 \mid X \setto 1)$, so that
$\Pr(Y=0 \mid X^*=0, X\setto 1) = 0$: a female who does not desire
treatment will always recover if treated.  Again, such a female should
be treated.

Finally we obtain exact identification marginally: $\pb= 0.28$.
Correspondingly, $\ph = \pb-\tau = 0$.  As there is no possibility of
harm, any female should be treated.

All the above conclusions agree with the DT prescription, based on the
experimental data alone: since $\ate > 0$, a female should be treated.
\end{example}

\begin{example}
  \label{ex:m}
  For males, the observational joint probabilities are as in
  \tabref{m}.
\begin{table}[h]
  \centering
  \begin{tabular}[c]{c|cc|c} 
    \multicolumn{1}{c}{} & {$Y=1$} & \multicolumn{1}{c}{$Y=0$}\\\hline
    $X =1$ & $0.49$ & $0.21$ & $0.70$\\
    $X = 0$ & $0.21$ & $0.09$ & $0.30$\\\hline
                         &  $0.70$ & $0.30$ & $1$\\
  \end{tabular}
  \caption{Joint observational distribution of $(X,Y)$ for males}   
  \label{tab:m}
\end{table}
Proceeding similarly to \exref{f}, we find that $\pb$ and $\ph$ are
again identified exactly: $\pb = 0.49$, $\ph = 0.21$.  Indeed,
$\Pr(Y=1, X=1) = 0.49 = \Pr(Y=1 \mid X\setto 1)$, implying the
deterministic property $\Pr(Y=1 \mid X^* = 0, X \setto 1) = 0$: a male
who does not desire treatment will never recover if treated.
Consequently such a male should not be treated.  Also,
$\Pr(Y=1, X=0)= 0.21 = \Pr(Y=1 \mid X \setto 0)$, so that
$\Pr(Y=1 \mid X^* = 1, X \setto 0) = 0$: a male who desires treatment
will never recover if untreated, so that such a male should be
treated.

However, if we do not observe $X^*$ for the target male patient, the
above does not tell us how to proceed.  We might try to balance $\ph$
($= 0.49$) and $\pb$ ($=0.21$) somehow: for example, treat just when
$\pb > \lambda \ph$ for some chosen value of $\lambda$.  In the light
of the clinical maxim {\em primum non nocere\/}, a value $\lambda = 3$
might be chosen---in which case the target male would not be
treated.\footnote{This argument parallels one in
  \citet{mueller/pearl:blog}, having different numbers, and
  $\lambda = 2$.}

By contrast, in the absence of knowledge of $X^*$ for the target male,
the DT approach would take no account of the observational data, again
focusing simply on $\ate = 0.28$---and so decide to treat.  In a large
population of similar cases, this would lead to an overall recovery
rate of $49\%$, the maximum possible; whereas the above strategy based
on balancing $\pb$ and $\ph$ would only have a $21\%$ recovery rate.
It is difficult to see how this could be regarded as ethical.
\end{example}

\begin{example}
  \label{ex:x}
  MP further consider another case.  Again, the interventional
  probabilities are $\Pr(Y=1 \mid X\setto 1) = 0.49$,
  $\Pr(Y=1 \mid X\setto 0) = 0.21$, with $\tau = 0.28$.  Now the
  observational joint probabilities are as in \tabref{x}.
\begin{table}[h]
  \centering
  \begin{tabular}[c]{c|cc|c} 
    \multicolumn{1}{c}{} & {$Y=1$} & \multicolumn{1}{c}{$Y=0$}\\\hline
    $X =1$ & $0.2$ & $0.5$ & $0.7$\\
    $X = 0$ & $0.1$ & $0.2$ & $0.3$\\\hline
                         &  $0.3$ & $0.7$ & $1$\\
  \end{tabular}
  \caption{Another joint observational distribution of $(X,Y)$}   
  \label{tab:x}
\end{table}
We compute
\begin{eqnarray*}
  0.7 \times \tau(1) &=& 0.09\\
  0.3 \times \tau(0) &=& 0.19\\
  0.7 \times \rho(1) &=& -0.39\\
  0.3 \times \rho(0) &=& 0.09.
\end{eqnarray*}
Using \eqref{impb} we find $0.28 \leq \xi \leq 0.52$, whence
$\pb = \half(\xi+\tau)$ is bounded between $0.28$ and $0.40$---and so
$\ph = \pb-\tau$ lies between $0$ and $0.12$.  So, even with the aid
of the additional observational data, we have not been able to
identify these probabilities exactly.  And even if we were to resolve
the ambiguity somehow, for example by taking the midpoints of these
intervals as suggested by \citet{li/pearl:ijcai19}, we
would be no better off than we were in \exref{m}, where trying to
balance $\pb$ against $\ph$ could lead to a decision opposite to the
rcommendation of the simple DT analysis, so leading to fewer
recoveries.
\end{example}


%% file: super2.tex
As seen in \secref{comb}, the real payoff from combining experimental and observational data is access to distributions involving the PT variable $X^*$, which is a sufficient covariate\footnote{In the case that $X^*$ has been observed in the experimental data, these distributions will be directly estimable, and observational data are not required.}.  In particular, we learn the conditional interventional response probabilities $\Pr(Y=1 \mid X^* =x^*, X\leftarrow x)$\footnote{If there are additional observed covariates, these should be conditioned on also.}, for all $(x,x^*)$.  These are exactly the relevant quantities for informing decision-making when $X^*$ is observed for the target patient.  This suggests a more straightforward approach to take all the background (experimental and observational) data into account: ask the patient (or their doctor) their preferred treatment, and use this information to inform, and so improve, the decision strategy.  This is the ``super-optimal'' regime of \citet{stensrud2022optimal}, who show that this straightforward DT-based approach will always be as good as a strategy based on considerations of benefit and harm.

In \exref{f} and \exref{m}, taking $X^*$ into account yields a perfect prediction of the response under one of the treatments, and the optimal decision then becomes obvious.

For \exref{x}, using \eqref{11}--\eqref{10} we obtain \begin{eqnarray*}
  \Pr(Y=1 \mid X^*=1, X\leftarrow 1) &=& 20/70\\
  \Pr(Y=1 \mid X^*=1, X\leftarrow 0) &=& 11/70\\
  \Pr(Y=1 \mid X^*=0, X\leftarrow 1) &=& 29/30\\
  \Pr(Y=1 \mid X^*=0, X\leftarrow 0) &=& 10/30.
\end{eqnarray*}
Consequently, whichever value of $X^*$ is observed, the optimal decision (maximising the conditional probability of recovery) is $X\leftarrow 1$, ``treat''.  In this case knowledge of $X^*$ has not changed the optimal decision, but in others cases it could.


%% file: prag2.tex
Here we identify and discuss some of the assumptions underlying the
foregoing analyses.

\subsection{Representative data}
\label{sec:rep}
A fundamental assumption underlying both the decision-theoretic
analysis of \secref{DT} and the alternative approach of
\secref{approach} is that the data available for estimating the
interventional probabilities $\Pr(Y=y, L=l \mid X\setto x)$ are on
individuals who can be regarded as ``similar to'' (``exchangeable
with'') the target case, so that these estimated probabilities are
applicable to the target.\footnote{For application to the MP arguments
  of \secref{approach}, the representativeness assumption should
  apparently be extended to the (typically unidentifiable) bivariate
  distribution, along with the other variables, of the pair of
  potential responses $(Y(1),Y(0))$.  For the interval-valued
  inferences made, however, this is not crucial, since these allow for
  arbitrary dependence in this bivariate distribution.}  In reality
this is implausible.  A clinical trial will have entry criteria and
processes that make its subjects quite untypical of the population
from which they are drawn, or indeed of the individuals recruited into
another such trial.  In any case, despite the name, entry criteria
govern who does {\em not\/} get into a trial: they cannot guarantee
that those who enter are representative even of a target individual
meeting the same criteria.

A clinical trial gains its value, not from representativeness, but
from the internal randomisation that ensures that a comparison between
its treated and untreated groups is indeed a comparision of like with
like, and that valid probability statements can be made about likely
differences, so enforcing internal validity.  Because of
unrepresentativeness it would not be appropriate to regard
$\Pr(Y=y, L=l \mid X\setto x)$, estimated from the data, as being
directly relevant to the target case---the problem of external
validity.  (One cheating way round this is to focus on a hypothetical
target individual who {\em can\/} be assumed exchangeable with those
in the study.)  Nevertheless, it may often be reasonable to regard the
estimated \ate\ or \cate\ as applying to the target---if not in its
exact numerical value, at least in its sign, which is all that is is
required, for DT application, to solve the single patient treatment
problem; or in its ordering of the $\cate_i$, as required to solve the
DT unit selection problem.

To underline how unreasonable the representative assumption is, it
should be noted that even when clinical trials with similar protocols
are compared this assumption is not made.  A striking example of its
failure for nearly identical protocols is given by the TARGET study
\citep{senn:target}, in which osteoarthritis patients in some centres
were randomised to receive either lumiracoxib or naproxen, and
patients in other centres either lumiracoxib or ibuprofen.  The degree
of comparability in design of the two sub-studies thus defined was
greater than one would typically expect between two randomised
controlled trials (RCTs), and {\em a fortiori\/} than between an RCT
and an observational study, such as MP consider.  Nevertheless, very
important differences at baseline were seen between the two
sub-studies, even though within-sub-study treatment arms were
comparable.  Furthermore, it was possible to demonstrate differences
at outcome between the two studies using lumiracoxib data only, a
striking illustration of a study effect.  It is generally accepted by
sponsors and regulators that as soon as concurrent control is
abandoned the greatest of care must be taken in drawing inferences.
Modern work on using data on historical controls to try and improve
the efficiency of clinical trials takes such study-to-study variation
as a given that must be allowed for \citep{schmidli}.

\subsection{Combination of data}
\label{sec:comb2}
An essential requirement for the application of \thmref{dawid} is that
the observational and experimental datasets comprise similar
individuals, so that the same probabilities for $X, X^*, Y$ apply to
both groups.  This is even more implausible than the
representativeness of either group.  In particular, the assumption of
a common distribution for the preferred treatment $X^*$, in both the
experimental and observational datasets and in the target patient, is
vital but highly questionable.  For a trial to be ethical, there
should be a degree of equipoise, with no strong belief about the
difference between the treatments.  A patient who has a strong
preference (our $X^*$) for one of the treatments is unlikely to agree,
or to be allowed, to be enrolled in the trial, so biasing the
distribution of $X^*$ in the trial.  So even if we were to accept the
mathematical arguments of MP based on combining observational and
experimental data, without this distributional invariance property
they are simply irrelevant.

\subsection{What do clinical trialists do in practice?}
\label{sec:clin}
The key to using RCTs is to identify reasonable assumptions, and use
theory to transfer results from trial to practice.  A striking example
is given by bioequivalence studies.  The subjects are usually young
healthy volunteers, frequently male.  However, the results will be
used to decide on appropriate treatments for elderly frail patients,
some female.  There is no pretence of representativeness.  Instead,
tight control and sensible scales of analysis are used.  The purpose
of such studies is to compare two formulations in terms of
bioavailability, and this is typically done using a cross-over trial
in which each subject is their own control, the order of
administration being randomised.  On separate days, concentration of
the test and reference pharmaceuticals are measured, and the ratio of
the areas under the two concentration time curves (AUCs) is calculated
for each subject, then analysed over all subjects, typically after
log-transformation.  What is relevant for treating an individual
patient is their own AUC: too low and efficacy may be disappointing,
too high and the drug may not be tolerated.  However, no inference is
made from a bioequivalence study in terms of AUCs alone, since they
would be quite different in healthy volunteers and patients.  Instead,
the idea is that the ratio between test and reference ought to be the
same in volunteers and patients, and this ratio can be used to make
predictions as to how the test drug will behave in clinical practice.
An interesting example of such a study is reported by
\citet{shumaker}.  They used a more elaborate design in which test and
reference drugs were given in a double cross-over, thus permitting
them to analyse the formulation-by-subject interaction.  They were
able to demonstrate that there was no evidence of an individual
bioequivalence effect: although you could estimate the individual
relative bioavailability, using the average over all subjects would be
superior than any such naïve estimate.  This raises a further issue
with MP, who assume that individual causal effects are stable over
time.  Moreover, causal identification analysis, such as conducted
here, essentially assumes an infinite sample size---but no infinities
are available for individual subjects, and estimating individual
causal effects requires close attention to components of variance.
Bioequivalence studies are an extreme example, but the general idea of
transferring results using a suitable scale for analysis, and
back-transforming to a scale suitable for decision analysis, is
commonplace: see \citet{lubsen} for a general discussion and
\citet{senn:covid} in the specific context of vaccine efficacy.  Of
course, as the COVID-19 pandemic has reminded us, there are no
guarantees.  Things that work at one time may not do so at another.
It behoves all those proposing solutions to be cautious and humble.


%% file: disc2.tex
We have given careful accounts of the DT and MP approaches to
individualised treatment choice.  The DT approach is simple in the
extreme, and selects the treatment strategy that maximises the number
of recoveries.  In contrast, the MP approach fixates on
philosophically questionable and unknowable counterfactual concepts,
and when its recommendations differ from those of DT will lead to
fewer recoveries.  This has been illustrated in a number of examples.

One feature of the MP approach is the combination of experimental and
observational data.  When some very strong conditions are satisfied,
this permits identification of the distribution of a special
covariate, the preferred treatment (PT).  As with any other covariate
whose distribution is known, this can then feed back to tighten the MP
interval inferences.

However---in those rare cases where the strong requisite assumptions
may be acceptable---there is a better approach.  We learn from the
combination of experimental and observational data that preferred
treatment, PT, is a covariate that, if observed prior to the point of
decision, could improve decision-making: formulae
\eqref{11}--\eqref{01} supply the information needed to conduct this.
In particular we have shown that, in just those very special cases
that use of PT leads to point identification of the MP probabilities
of benefit and of harm, knowledge of the target patient's PT value
allows perfect prediction of the outcome under at least one of the
treatment interventions, and so to a trivial solution to the decision
problem.  Since this particular variable should be readily available
(though it may be difficult to extract in an unbiased fashion), this
is an easy win. \citet{mats:rejoinder} discuss the
relationship between this ``superoptimal'' decision-theoretic
approach, and the approach of MP.  We argue that this would be a much
better use of the combined data.

The DT approach has a long history of fruitful application to an
enormous variety of fields, from clinical trials to rocket science.
Attempts to replace it with another approach, based on
counterfactuals, are totally unnecessary and dangerously misguided.
Such an approach should not be applied in practice.
